\documentclass[letterpaper, 10 pt, conference]{ieeeconf}  

\IEEEoverridecommandlockouts                              

\usepackage{cite}
\usepackage{amsmath,amssymb,amsfonts}
\usepackage{graphicx} 
\usepackage{amsmath}
\usepackage{amssymb}  
\usepackage{mathrsfs}
\usepackage{multicol}
\usepackage[dvipsnames]{xcolor}
\usepackage{float}
\usepackage{subcaption}
\usepackage{mathtools}
\usepackage{soul}
\usepackage{bm}
\usepackage[T1]{fontenc}
\newtheorem{remark}{Remark}

\newtheorem{definition}{Definition}
\newtheorem{proposition}{Proposition}

\newtheorem{theorem}{Theorem}

\usepackage{tikz}
\usepackage{xcolor}
\usepackage{dsfont}
\usepackage{float}
\usepackage{algorithm}
\usepackage{amsmath}
\usepackage{tcolorbox}

\usepackage{booktabs}
\usepackage{pifont}

\usepackage{algorithm}
\usepackage{algpseudocode}

\usepackage{array}

\usepackage{algpseudocode}
\usepackage{hyperref} 

\let\labelindent\relax
\usepackage[shortlabels]{enumitem}

\DeclareMathAlphabet\mathbfcal{OMS}{cmsy}{b}{n}

\definecolor{darkspringgreen}{rgb}{0.09, 0.45, 0.27}

\def\BibTeX{{\rm B\kern-.05em{\sc i\kern-.025em b}\kern-.08em
    T\kern-.1667em\lower.7ex\hbox{E}\kern-.125emX}}

\newcommand{\hatbf}[1]{\widehat{\mathbf{#1}}}
\newcommand{\norm}[1]{\left\lVert#1\right\rVert}

\newcommand{\cmark}{\textcolor{darkspringgreen}{\ding{51}}}%
\newcommand{\xmark}{\textcolor{red}{\ding{55}}}%

\def\1{\mathds{1}}

\title{\LARGE \bf
MAD: A Magnitude And Direction Policy Parametrization for Stability Constrained Reinforcement Learning
}

\author{Luca Furieri, Sucheth Shenoy, Danilo Saccani, Andrea Martin, and Giancarlo Ferrari-Trecate%
\thanks{This project was conceived and led by Luca Furieri as Principal Investigator under the SNSF Ambizione grant PZ00P2\_208951. Sucheth Shenoy, Danilo Saccani, Andrea Martin, and Giancarlo Ferrari-Trecate acknowledge financial support from the Swiss National Science Foundation through NCCR Automation (grant 51NF40\_225155) and the NECON project (grant 200021\_219431).}%
\thanks{All authors were affiliated with the Institute of Mechanical Engineering, École Polytechnique Fédérale de Lausanne (EPFL) when this research was carried out. Current affiliations: L. Furieri, Department of Engineering Science, University of Oxford, United Kingdom (\texttt{luca.furieri@eng.ox.ac.uk}); S. Shenoy, Institute for Data Science Foundations, Hamburg University of Technology (TUHH), Germany (\texttt{sucheth.shenoy@tuhh.de}); A. Martin, School of Electrical Engineering and Computer Science and Digital Futures, KTH Royal Institute of Technology, Sweden (\texttt{andrmar@kth.se}); D. Saccani and G. Ferrari-Trecate, Institute of Mechanical Engineering, EPFL, CH-1015 Lausanne, Switzerland (\texttt{\{danilo.saccani, giancarlo.ferraritrecate\}@epfl.ch}).}}

\begin{document}

\maketitle
\thispagestyle{empty}
\pagestyle{empty}


\begin{abstract}
We introduce \emph{magnitude and direction} (MAD) policies, a policy parameterization for reinforcement learning (RL) that preserves $\ell_p$ closed-loop stability for nonlinear dynamical systems. Despite their completeness in describing all stabilizing controllers,  methods based on nonlinear Youla and system-level synthesis are significantly impacted by the difficulty of parametrizing $\ell_p$-stable operators. In contrast, MAD policies introduce explicit feedback on state-dependent features – a key element behind the success of reinforcement learning pipelines – without jeopardizing closed-loop stability. This is achieved by letting the magnitude of the control input be described by a disturbance-feedback $\ell_p$-stable operator, while selecting its direction based on state-dependent features through a universal function approximator. We further characterize the robust stability properties of MAD policies under model mismatch. Unlike existing disturbance-feedback policy parametrizations, MAD policies introduce state-feedback components compatible with model-free RL pipelines, ensuring closed-loop stability with no model information beyond assuming open-loop stability.  Numerical experiments show that MAD policies trained with deep deterministic policy gradient (DDPG) methods generalize to unseen scenarios – matching the performance of standard neural network policies while guaranteeing closed-loop stability by design.\looseness-1

\end{abstract}

\section{Introduction}


As machine learning algorithms become increasingly intertwined with the dynamics of modern cyber-physical systems, ensuring fail-safe and optimal operation remains a fundamental challenge. Reinforcement learning (RL) has demonstrated remarkable success in software-based domains, such as recommender systems and competitive games \cite{silver2016mastering}, and has recently been applied to safety-critical hardware tasks, including legged locomotion on rough terrain \cite{lee2020learning} and high-speed drone racing \cite{kaufmann2023champion}. At the core of these successes lies a key separation of concerns: in many applications, domain-specific engineering principles are still essential to ensure stability and safety, while RL is typically layered on top to optimize performance – often without formal guarantees, especially when interacting with complex, uncertain environments.

In RL for continuous control tasks, policy optimization typically relies on actor-critic methods \cite{konda1999actor}, where the actor learns a policy and the critic estimates its value. Deep deterministic policy gradient (DDPG) \cite{lillicrap2015continuous} embraces the deterministic policy gradient theorem \cite{silver2014deterministic} and leverages deep neural networks and experience replay. Soft actor-critic \cite{haarnoja2018soft} further enhances robustness by incorporating entropy-regularized learning, making it well-suited for high-dimensional control tasks. However, despite these advancements, the stability guarantees of RL methods as applied to dynamical systems are still not fully understood. Fundamental theoretical limitations hinder the direct optimization of control policies within feedback loops; for instance, \cite{fazel2018global} demonstrated that even in simple linear-quadratic regulator settings, standard policy-gradient algorithms such as REINFORCE \cite{williams1992simple} can drive a system unstable unless the step size is carefully based on a case-by-case analysis. This is also the case when using actor-critic methods, which introduce additional challenges such as slow convergence \cite{yang2019provably}.\looseness-1

Towards establishing guarantees for control applications, one line of work develops RL frameworks that embed structural constraints to ensure closed-loop stability. For instance, \cite{shi2022stability} shows that strict policy monotonicity is sufficient for voltage control with embedded stability guarantees, constructing an explicit neural network Lyapunov function and proposing a stable-DDPG algorithm for policy training. Another class of methods combines RL with model predictive control (MPC), leveraging MPC’s inherent stability and safety properties while improving performance through learning \cite{zanon2020safe}. 
Finally,~\cite{roberts2011feedback, furierineuralSLS, furieri2024learningtoboost, lawrence2024stabilizing} generalize classical control frameworks (e.g., Youla, SLS) to nonlinear systems using explicit models of pre-stabilized dynamics. For convenience, we collectively refer to these techniques as disturbance-feedback (DF) methods.

While DF methods \cite{roberts2011feedback, furierineuralSLS, furieri2024learningtoboost, lawrence2024stabilizing} provide a complete characterization of all and only the stabilizing control policy, a major practical limitation is that their expressivity relies entirely on the ability to parameterize dynamic operators that are $\ell_p$-stable. Constructing rich and expressive parametrizations of the space of $\ell_p$-stable operators remains a fundamental bottleneck: no known universal function approximator can represent all (and only) $\ell_p$-stable operators as the number of parameters increases. Another limitation is that the DF methods proposed in \cite{roberts2011feedback, furierineuralSLS, furieri2024learningtoboost} rely on full or partial knowledge of the system model to reconstruct disturbances, which are then used to design an optimal closed-loop response. In other words, these approaches refrain policies from reacting directly to state measurements – a key feature in the success of RL – as this may compromise stability. Furthermore, when the model is entirely unknown and disturbances cannot be reconstructed, only stable feedforward terms can be applied within the DF method of \cite{furieri2024learningtoboost}. To make RL a viable framework for nonlinear optimal control with formal guarantees, policy parameterizations must allow state-dependent responses – even when the system dynamics are fully unknown.
\subsubsection*{Contributions}  
This paper introduces \emph{magnitude and direction (MAD) policies}, a policy parametrization that makes DF methods compatible with standard RL pipelines for the optimal control of pre-stabilized nonlinear systems, enabling the inclusion of explicit state-dependent components while preserving closed-loop stability guarantees. This is achieved through a structured decomposition of the control policy into two terms: a state-feedback direction term, which is freely parameterized (e.g., via neural networks), and an $\ell_p$-stable magnitude term, which leverages available approximations of the space of $\ell_p$-stable operators.

We first prove that, 
MAD policies expand the set of achievable closed-loop behaviors compared to DF methods, owing to the introduction of freely chosen state-feedback direction terms. When the system model is known, MAD policies parametrize all stabilizing controllers. Finally, we characterize the robustness of MAD policies under model mismatch, deriving conditions that preserve stability as the mismatch increases. In contrast to DF methods – which, under complete model uncertainty reduce to applying fixed stable feedforward terms and are thus incompatible with model-free RL – MAD policies maintain the ability to search over state-feedback policies while preserving closed-loop stability. Numerical experiments demonstrate that MAD policies trained through DDPG showcase generalization capabilities comparable to those of standard RL policies.

\emph{Notation: }The set of all sequences $\mathbf{x} = (x_0,x_1,x_2,\ldots)$, where $x_t \in \mathbb{R}^n$ for all $t\in \mathbb{N}$ is denoted as $\ell^n$. Moreover,  $\mathbf{x}$ belongs to $\ell_p^n \subset \ell^n$ with $p \in \mathbb{N}$ if ${\mathbf{x}}_p = \left(\sum_{t=0}^\infty |x_t|^p\right)^{\frac{1}{p}} < \infty$, where $|\cdot|$ denotes any vector norm. We say that $\mathbf{x} \in \ell^n_\infty$ if $\operatorname{sup}_{t}|x_t|< \infty$. When clear from the context, we omit the superscript from $\ell^n$ (resp. $\ell^n_p$) and write $\ell$ (resp. $\ell_p$).
We use the notation $x_{j:i}$ to refer to the truncation of $\mathbf{x}$ to the finite-dimensional vector $(x_i,x_{i+1},\ldots,x_{j})$. An operator $\mathbf{A}:\ell^n \rightarrow \ell^m$ is said to be \emph{causal} if $\mathbf{A}(\mathbf{x}) = (A_0(x_0),A_1(x_{1:0}),\ldots,A_t(x_{t:0}),\ldots)$. If in addition $A_t(x_{t:0}) = A_t(0,x_{t-1:0})$, then $\mathbf{A}$ is said to be strictly causal. Similarly, $A_{j:i}(x_{j:0}) =  (A_i(x_{i:0}),A_{i+1}(x_{i+1:0}),\ldots,A_j(x_{j:0}))$. 
An operator $\mathbf{A}:\ell^n\rightarrow \ell^m$ is said to be $\ell_p$-stable if it is causal and $\mathbf{A}(\mathbf{a}) \in \ell_{p}^m$ for all $\mathbf{a} \in \ell_{p}^n$. Equivalently, we  write $\mathbf{A} \in \mathcal{L}_{p}$.  We say that an $\mathcal{L}_p$ operator $\mathbf{A}:\mathbf{w}\mapsto \mathbf{u}$  has finite $\mathcal{L}_p$-gain $\gamma(\mathbf{A})>0$ if  $\|\mathbf{u}\|_p\leq \gamma(\mathbf{A})\|\mathbf{w}\|_p$, for all $\mathbf{w}\in\ell_p^n$.

\section{Problem formulation}

We consider open-loop stable, or prestabilized, nonlinear discrete-time systems
\begin{equation}
    \label{eq:nonlinsys}
    x_{t} = f(x_{t-1},u_{t-1})+w_t\,,~~~t= 1,2,\ldots\,,
\end{equation}
where $x_t \in \mathbb{R}^n$ denotes the state vector, $u_t \in \mathbb{R}^m$ the control input, and $w_t \in \mathbb{R}^n$ the unknown process noise. In \emph{operator form}, \eqref{eq:nonlinsys} is rewritten as
\begin{equation}
\label{eq:operator_form}
    \mathbf{x} = \mathbf{F}(\mathbf{x},\mathbf{u}) + \mathbf{w}\,,
\end{equation}
where $\mathbf{F}:\ell^n\times \ell^m \rightarrow \ell^n$ is a strictly causal operator defined as $\mathbf{F}(\mathbf{x},\mathbf{u}) = (0,f(x_0,u_0),\ldots,f(x_{t},u_{t}),\ldots)$. We assume  that the operator
$ \mathbfcal{F}:(\mathbf{u},\mathbf{w})\mapsto \mathbf{x}\,,
$ which is the unique causal operator mapping a sequence of inputs and disturbances $(\mathbf{u}, \mathbf{w})$ to the corresponding state trajectory $\mathbf{x}$, belongs to $\mathcal{L}_p$. This implies that the system is open-loop $\ell_p$-stable or has been pre-stabilized. As a result, $\mathbf{x} \in \ell_p$ whenever $\mathbf{u} \in \ell_p$ and $\mathbf{w} \in \ell_p$.

However, in many engineering applications ensuring stability is not sufficient \cite{annaswamy2024control}; achieving satisfactory performance is also crucial. We consider an infinite-horizon discounted loss for a given control sequence $\mathbf{u}$ defined as
\begin{equation}
\label{eq:perf_definition}
L^\infty(\mathbf{x},\mathbf{u})=\mathbb{E}_{w_t\sim \mathcal{D}} \left[ \sum_{t=0}^{\infty}\alpha^tl(x_t,u_t) \right]\,,
\end{equation}
where  $l:\mathbb{R}^n\times \mathbb{R}^m \to \mathbb{R}$ is a continuous loss function such that, for all $\mathbf{u} \in \ell_p$ and every $0<\alpha<1$, it holds that $L^\infty(\mathbf{x},\mathbf{u})< \infty$, that is, the infinite-horizon loss is finite whenever the control input is a stable signal within $\ell_p$. 
Note that, in order to define the expectation in \eqref{eq:perf_definition}, the noise $w_t$ is assumed to follow an unknown distribution $\mathcal{D}$.

\subsection{Stability constrained reinforcement learning}

As mentioned above, our goal is to determine an optimal sequence of control inputs $ \mathbf{u}^\star $ that minimizes the infinite-horizon expected loss $ L^\infty(\mathbf{x},\mathbf{u}) $ while ensuring that the obtained policy guarantees closed-loop stability. Given the assumption that $ \mathbfcal{F} \in \mathcal{L}_p $,  this requirement translates into the condition that $ \mathbf{u} \in \ell_p $. We  formulate this \emph{stability-constrained RL problem} as  follows: 
\begin{equation}
    \min_{\mathbf{u} \in \ell_p,~\mathbf{x} =\mathbfcal{F}(\mathbf{u},\mathbf{w})}  L^\infty(\mathbf{x},\mathbf{u}).
    \label{RL_control_Problem}
\end{equation}
In other words, the system interacts with an environment subject to unknown disturbances, and the goal is to search over control input sequences that minimize the expected long-term cost while ensuring closed-loop stability in the sense
\begin{equation}
\label{eq:closed_loop_stability}
    \mathbf{w} \mapsto (\mathbf{x},\mathbf{u}) \in \mathcal{L}_p\,.
\end{equation}

Since the initial states and disturbances are not known in advance, the optimal sequence of control inputs $ \mathbf{u}^\star $ must be computed as a feedback law based on realized states, i.e., $ \mathbf{u}^\star = \mathbf{K}^\star(\mathbf{x}^\star) $, where $ \mathbf{K}^\star $ is the optimal causal policy and $\mathbf{x}^\star$ is the corresponding optimal state trajectory. A fundamental result in dynamic programming and RL states that, under standard regularity assumptions (bounded cost, compact input sets, and discount factor $ \alpha \in (0,1) $), the infimum over all causal policies can be achieved by a \emph{memoryless causal policy} of the form $ u_t^\star = \mu^\star(x_t^\star) $; we refer the interested reader to \cite{bertsekas2019reinforcement}. 



\emph{Policy training in RL:} To train the policy, many RL methods rely on policy gradient optimization. Deterministic policies\footnote{In RL, policies are often stochastic, meaning that $u_t$ is sampled from a learned distribution. However, deterministic policies offer more efficient policy gradient estimation for control tasks in continuous spaces \cite{lillicrap2015continuous}.} are parametrized, typically via neural networks, as $u_t=\mu(x_t,\theta)$, with parameters $\theta \in \mathbb{R}^d$. The parameters can then be updated iteratively, for example, through gradient descent:  $\theta\leftarrow \theta-\eta \nabla_\theta L^\infty(\theta)$, with $\eta>0$ sufficiently small. The deterministic policy gradient theorem \cite{silver2014deterministic} states that: 
\begin{equation} 
\label{eq:gradients_policy} \nabla_\theta L^\infty(\theta) = \mathbb{E}_{x\sim\rho^\mu}\left[\nabla_\theta \mu(x,\theta) \nabla_u Q_\mu(x,u)\big|{u = \mu(x,\theta)}\right]\,, 
\end{equation}
where $\rho^\mu$ denotes the state distribution induced by policy $\mu$, and the $Q$-function $Q_{\mu}$ satisfies the Bellman equation: 
\begin{equation} 
\label{eq:Bellman} 
Q_\mu(x,u) = l(x,u)+\alpha\mathbb{E}_{w \sim \mathcal{D}}\left[Q_\mu(x^+,\mu(x^+,\theta))\right]\,, 
\end{equation}
where $ x^+ = f(x,u)+w$ and $w\sim \mathcal{D}$. Actor-critic methods, such as DDPG \cite{lillicrap2015continuous}, leverage the structure of \eqref{eq:gradients_policy} and \eqref{eq:Bellman} to approximate the gradient efficiently. They parametrize the $Q$-function (the ``critic'') using a neural network $Q(x,u,\phi)$ and the policy (the ``actor'') using another neural network $\mu(x,\theta)$. During training,  the critic and the actor are iteratively updated to minimize the mismatches in  the policy gradient in \eqref{eq:gradients_policy} and the Bellman equation~\eqref{eq:Bellman}.

\smallskip

Crucially, restricting the search over static policies poses a challenge to the stability requirements. A randomly initialized memoryless policy $ u_t = \mu(x_t, \theta) $ may induce an unstable closed-loop map. This creates a fundamental difficulty for RL-based control: intermediate training stages may result in unstable closed-loop responses, potentially disrupting the learning process entirely.  This motivates extending the search to dynamic policies.\footnote{Dynamic policies induce a non-Markovian closed-loop behavior; this makes the expressions \eqref{eq:gradients_policy}-\eqref{eq:Bellman} invalid over dynamic policies. To remedy this issue, the literature of RL over memory-based policies appropriately augments the state vector $s$ to include the internal state of the policy. For implementation (Section~\ref{se:Numerical}), we adopt the solution proposed in  \cite{heess2015memory}. } We recall from~\cite{furieri2024learningtoboost} the DF characterization of all the causal control policies that guarantee closed-loop stability for the considered setup.

\begin{theorem}{(adapted from \cite{furieri2024learningtoboost})} \label{thm:neurSLS}
    Let $\mathbfcal{F} \in \mathcal{L}_p$, and define the control input as
    \begin{equation} \label{eq:dist_feedback}
        \mathbf{u}=\mathbfcal{M}(\mathbf{x-F(x,u)}) = \mathbfcal{M}(\mathbf{w}),
    \end{equation}
    where $\mathbfcal{M}:\ell^n\rightarrow \ell^m$ is a causal operator. If $\mathbfcal{M} \in \mathcal{L}_p$, then the closed-loop system satisfies \eqref{eq:closed_loop_stability}. Conversely, if a causal policy $\mathbf{u}=\mathbf{K}(\mathbf{x})$ ensures \eqref{eq:closed_loop_stability}, then there exists an operator $\mathbfcal{M}\in \mathcal{L}_p$ that achieves the same closed-loop behavior.
\end{theorem}

In line with classical results about nonlinear Youla parametrizations \cite{fujimoto2000}, Theorem~\ref{thm:neurSLS} provides a necessary condition for ensuring closed-loop stability, and highlights that exploring the space of operators $\mathbfcal{M} \in \mathcal{L}_p$ is sufficient to characterize all stabilizing policies. The key property is that, for any choice of $\mathbfcal{M} \in \mathcal{L}_p$, the corresponding control policy \eqref{eq:dist_feedback} is stabilizing. This completely decouples the problem of ensuring stability to that of optimizing over operators in $\mathcal{L}_p$ to solve \eqref{RL_control_Problem}. Accordingly, the stability-constrained RL problem \eqref{RL_control_Problem} admits the following reformulation over the class of all DF policies. 
\begin{subequations} 
    \begin{align}
        &\min_{\mathbf{x},\mathbf{u}}  && L^\infty(\mathbf{x},\mathbf{u})\label{eq:s-RL}\\
        & \operatorname{subject~to} && \mathbf{x} = \mathbfcal{F}(\mathbf{u},\mathbf{w})\,, \mathbf{u} = \mathbfcal{M}(\mathbf{x}-\mathbf{F}(\mathbf{x},\mathbf{u}))\,,\\
        &~&&\mathbfcal{M} \in \mathcal{L}_p\,. \label{eq:s-RL_end}    \end{align}
    \end{subequations}

\noindent Despite the generality of the policy parametrization established in Theorem~\ref{thm:neurSLS}, deploying RL routines for learning over DF policies $\mathbf{u} = \mathbfcal{M}(\mathbf{w})$ comes with novel challenges that we define in the following subsection.

\subsection{Open challenges of stability-constrained RL}

Previous work has proposed addressing stability-constrained RL \cite{lawrence2024stabilizing} and finite-horizon optimal control problems \cite{furieri2024learningtoboost, furierineuralSLS} by learning over DF policies of the form $\mathbf{u} = \mathbfcal{M}(\theta)(\mathbf{w})$, where
\begin{equation} \label{eq:class_parametrized}
\mathbfcal{M}(\theta) \in \mathcal{U} \subset \mathcal{L}_p, \quad \forall \theta \in \mathbb{R}^d\,.
\end{equation}
Here, $\mathcal{U}$ denotes a finite-dimensional class of $\mathcal{L}_p$ operators that can be explicitly characterized – for instance, as in \cite{orvieto2023resurrecting,lawrence2024stabilizing,wang2022learning}. However, no universal function approximator is known to capture all (and only) elements of $\mathcal{L}_p$, and non-asymptotic results on how well $\mathcal{U}$ covers $\mathcal{L}_p$ are not available. This limitation can significantly impact the efficacy of DF policies, as their expressivity entirely depends on how accurately $\mathcal{U}$ approximates $\mathcal{L}_p$.

The second and more fundamental challenge arises from the policy architecture \eqref{eq:dist_feedback} itself: computing control inputs based on \emph{external disturbances} can affect the generalization capabilities achievable through RL algorithms, whose efficacy has been showcased with policies that directly respond to state measurements $ x_t $.

%

\smallskip

\emph{Example:} 
Consider a state-feedback policy $u_t = \mu(x_t)$, such as a proportional controller $u_t = Kx_t$, and a realization of external disturbances $\mathbf{w} = (x_0, w_1, \ldots)$. Let the system \eqref{eq:nonlinsys} evolve under the policy $u_t = Kx_t$. The DF policy that reproduces the same closed-loop behavior can be expressed as $u_t = K\mathcal{F}_t(u_{t-1:0}, w_{t-1:0})$, introducing a recursive, nonlinear dependency in computing $u_t$. Recovering this simple behavior during learning may prove difficult, as a DF policy has to ``reverse engineer'' how disturbances influence the closed-loop state evolution through a nonlinear model. As a result, ensuring that even basic state-feedback policies are included in the class parameterized by DF policies may be challenging. 


In what follows, we develop a policy parametrization that mitigates the open challenges above. 



\section{MAD policies for stability-constrained RL}
We introduce a novel class of dynamic policies for stability-constrained RL of $\ell_p$ pre-stabilized dynamical systems. Through a polar  decomposition of the control input into magnitude and direction terms, the proposed class of MAD policies  enables behaviors that react directly to state measurements without requiring richer parametrizations of $\mathcal{L}_p$.\looseness-1

First, we show that when the system model \eqref{eq:nonlinsys} is known, MAD policies expand the range of achievable behaviors as compared to DF policies of the form \eqref{eq:dist_feedback}. Moreover, when $\mathbfcal{M}$ is allowed to be any $\mathcal{L}_p$ operator, MAD policies recover all and only the stabilizing control policies. However, in most RL-based control applications, the system model is partially or fully unknown. We show that MAD policies remain stabilizing even in such settings, with expressivity degrading as the model uncertainty increases. Yet, MAD policies retain expressive state-feedback behaviors even in the fully unknown case, making them compatible with model-free RL algorithms.

\subsection{Stability-constrained RL with model knowledge}
We here  define the proposed policy class assuming knowledge of the system model \eqref{eq:nonlinsys}. 
\begin{tcolorbox}[colback=gray!1.5, colframe=black!20]
\begin{definition}[MAD policies]
    We say that a policy is MAD for $\mathcal{U}$ if it can be written as
    \begin{align}
    \label{eq:MAD_policy}
        u_t = |\mathcal{M}_t(w_{t:0})+a_t(x_0)|\cdot D_t(x_{t:0})\,,
    \end{align}
    where we have:
    \begin{enumerate}
        \item a stable magnitude operator $\mathbfcal{M} \in \mathcal{U}\subseteq \mathcal{L}_p$ acting on disturbances $w_t=x_t-f(x_{t-1},u_{t-1})$,
        \item a feed-forward term $a_t(x_0)$ with $\bm{a}(x_0) \in \ell_p$ for every $x_0 \in \mathbb{R}^n$,
        \item a direction vector $D_t$ with $|D_t(x_{t:0})|\leq 1$.
    \end{enumerate}
\end{definition}
\end{tcolorbox}
The MAD policy class relies on a polar decomposition of the control input to enforce stable closed-loop behavior. The magnitude term $|\mathcal{M}_t(w_{t:0})+a_t(x_0)|$ adjusts the control effort based on the initial condition $x_0$ and the reconstructed disturbances $w_t$, while the direction term $D_t(x_{t:0})$ provides state-dependent feedback, subject to a global norm bound.


\begin{remark}
While MAD policies and DF methods both rely on a finite-dimensional approximation of $\mathcal{L}_p$-stable operators, MAD policies fundamentally differ in how they leverage this approximation. In DF methods, the entire closed-loop behavior is constrained by the expressivity of the $\mathcal{L}_p$ parametrization. In contrast, MAD policies separate stability and expressivity, by isolating the stability constraint within the magnitude term, while allowing the direction term to be freely parameterized (e.g., via neural networks). 
\end{remark}

For a fixed approximation $\mathcal{U} \subset \mathcal{L}_p$, policies that are MAD~for~$\mathcal{U}$ achieve a strictly richer set of stable behaviors than DF policies in the form \eqref{eq:dist_feedback}. To proceed with this analysis, define the set of behaviors 
    \begin{equation*}
        \mathcal{B}^{\mathtt{MAD}}_{\mathcal{U}}(\mathbf{w})=\{\mathbf{u} \in \ell: \exists\mathbfcal{M}\in \mathcal{U}: u_t\text{ complies with \eqref{eq:MAD_policy}} \}\,,
    \end{equation*}
that are MAD for $\mathcal{U}$, and the set of DF behaviors  \cite{furieri2024learningtoboost} as
    \begin{equation*}
        \mathcal{B}_{\mathcal{U}}(\mathbf{w})=\{\mathbf{u} \in \ell: \exists\mathbfcal{M}\in \mathcal{U}: \mathbf{u} = \mathbfcal{M}(\mathbf{w}) \}\,.
    \end{equation*}
\noindent We next show that, when the system model \eqref{eq:nonlinsys} is known, MAD policies enable searching over an extended set of behaviors. Further, policies that are MAD for $\mathcal{L}_p$ comprise the set of all stabilizing policies for the system \eqref{eq:nonlinsys}.

\begin{theorem}
\label{th:extended-param}
Let $\mathcal{U}$ be a set of  operators. 
For every $\mathbf{w}\in\ell$ it holds that
\begin{equation*}
   \mathcal{B}_{\mathcal{U}}(\mathbf{w}) 
\;\subseteq\;
\mathcal{B}_{\mathcal{U}}^{\mathtt{MAD}}(\mathbf{w})\,.
\end{equation*}
Furthermore, if $\mathcal{U} \in \mathcal{L}_p$ and $\mathbf{w} \in \ell_p$, then  $\mathcal{B}_{\mathcal{U}}^{\mathtt{MAD}}(\mathbf{w}) \in \ell_p$. Last, a policy $\mathbf{u} = \mathbf{K}(\mathbf{x})$ is $\ell_p$-stabilizing if and only if it is MAD for $\mathcal{L}_p$.
\end{theorem}
\begin{proof}
    The proof can be found in the Appendix~\ref{app:proof_theorem}
\end{proof}

In order for MAD policies to span the set of \emph{all} stabilizing control policies as per Theorem~\ref{th:extended-param}, perfect knowledge of the system model \eqref{eq:nonlinsys} is required.  Specifically, achieving a target stable closed-loop behavior \(\mathbf{w} \mapsto (\bar{\mathbf{x}}, \bar{\mathbf{u}})\) requires setting \(\mathbfcal{M} = \bar{\bm{\Psi}}\), where \(\bar{\bm{\Psi}}\) is the desired mapping. Since \(\bar{\mathbf{u}} = \mathbfcal{M}(\bar{\mathbf{x}} - \mathbf{F}(\bar{\mathbf{x}}, \bar{\mathbf{u}}))\) depends on the system dynamics \(\mathbf{F}\), it can only be computed with full model knowledge. Next, we deal with characterizing MAD policies that remain stabilizing in the presence of a partially known or completely unknown system model.

\subsection{Stability-constrained RL without model knowledge}

Let us denote the nominal model available for design as ${\hatbf{F}}(\mathbf{x},\mathbf{u})$ and the real unknown plant as 
\begin{equation}
    \label{eq:mismatch_plant}
    \mathbf{F}(\mathbf{x},\mathbf{u}) = \hatbf{F}(\mathbf{x},\mathbf{u}) + \bm{\Delta}(\mathbf{x},\mathbf{u})\,,
\end{equation}
where $\bm{\Delta}$ is a strictly causal operator representing the model mismatch. We say that we have partial knowledge about $\mathbf{F}$ if  $\bm{\Delta}$ has a finite $\ell_p$ gain, that is, $\gamma(\bm\Delta)<\infty$. Otherwise, we assume that $\gamma(\bm{\Delta}) = \infty$. Given a nominal system model $\widehat{\mathbf{F}}$, it is not possible anymore to exactly reconstruct external disturbances from state measurements – thus preventing the application of Theorem~\ref{th:extended-param}. Motivated as such, we study the robust stabilization properties of MAD policies in the presence of a partially or fully unknown system model. 

\begin{proposition}
	\label{th:result_robust} 
	Let $\gamma(\mathbf{\Delta})$ denote the $\ell_p$ gain of the operator  $\bm{\Delta}$ in \eqref{eq:mismatch_plant}. Furthermore, assume that
 the operator $\mathbfcal{F}$  has finite $\ell_p$-gain $\gamma(\mathbf{\mathbfcal{F}})$.  Then, for any $\mathbfcal{M}$ such that 
 \begin{equation}
 \label{eq:condition_robustness}
 \gamma(\mathbfcal{M})<\gamma(\mathbf{\Delta})^{-1}(\gamma(\mathbfcal{F})+1)^{-1}\,,
 \end{equation}
 for any $\bm{a}(x_0)\in \ell_p$, and for any direction term with $|D_t(\cdot)|\leq 1$, the MAD control policy given by
 \begin{subequations}
 \begin{align}
 &\widehat{w}_t = x_t - \hat{f}(x_{t-1},u_{t-1})\,,\nonumber\\
 &u_t = |\mathcal{M}_t(\widehat{w}_{t:0})+a_t(x_0)|\cdot D_t(x_{t:0})\,, \nonumber
 \end{align}
 \end{subequations}
 stabilizes the system \eqref{eq:nonlinsys}. 
 \end{proposition}
 \begin{proof}
     The proof can be found in Appendix~\ref{app:proof_prop}.
\end{proof}


Proposition~\ref{th:result_robust} shows that when the model is perfectly known (i.e., $\gamma(\bm{\Delta})=0$), all stable closed-loop behaviors from Theorem~\ref{th:extended-param} are recovered, as $\widehat{\mathbf{w}} = \mathbf{w}$. As the gain $\gamma(\bm{\Delta})$ increases, disturbance reconstruction from state measurements becomes less accurate, requiring a sufficiently small $\ell_p$ gain for $\mathbfcal{M}$.  Crucially, MAD policies enable the design of stabilizing state-feedback terms even when $\gamma(\bm{\Delta}) \rightarrow \infty$; this is in contrast to DF policies, where accurate reconstruction of $\widehat{\mathbf{w}}$ is not possible without a model. In this case,  \eqref{eq:MAD_policy} simplifies to $u_t = |a_t(x_0)|D_t(x_{t:0})$, where $a_t(x_0)$ acts as a feed-forward term and $D_t(x_{t:0})$ is a neural state-feedback direction term.  As we show in the next section, this enables MAD policies to achieve generalization similar to standard RL policies while maintaining stability without model knowledge. We summarize the key features of MAD policies as compared to DF and standard RL policies in Table~\ref{tab:comparison}.\looseness-1




\begin{table*}[t]
    \centering
    \caption{Summary of key features of MAD policies in contrast with DF and standard RL policies.}
    \label{tab:comparison}
    \begin{tabular}{@{}l>{\arraybackslash}m{4.1cm}>{\arraybackslash}m{4.1cm}>{\arraybackslash}m{3.9cm}@{}}
        \toprule
        & \centering \textbf{MAD policies} & \centering \textbf{DF policies} & \qquad ~\textbf{standard RL policies}\\ 
        \midrule
        Guarantees for pre-stabilized systems & \cmark~all and only stable closed-loops & \cmark~all and only stable closed-loops& \xmark~ conservative assumptions\\
        Stability using model-free RL & \cmark~if unknown system is stable  & \xmark~cannot implement & \xmark~no a-priori guarantees \\
        Explicit neural state-feedback & \cmark~in the direction term & \xmark~implicitly through disturbances& \cmark~neural state-feedback\\
        Generalization capabilities & \cmark~enhanced by direction term & \xmark~low (no explicit state-feedback) & \cmark~highest\\
        \bottomrule
    \end{tabular}
    \vspace{-0.2cm}
\end{table*}
\section{Numerical examples}
\label{se:Numerical}
In this section, we discuss implementations of MAD policies and how they can be trained using off-the-shelf RL algorithms.  We consider the $\mathtt{corridor}$ benchmark example from \cite{furierineuralSLS} and train MAD policies through both model-based and model-free DDPG \cite{lillicrap2015continuous}; see  Appendix~\ref{app:mountains} for the definition of the system dynamics and the task loss function. We also train 1) DF policies $\mathbf{u}=\mathbfcal{M}(\mathbf{w})$ from \cite{furieri2024learningtoboost} and 2) standard memoryless RL policies $u_t = \mathtt{NN}(x_t)$ to assess the expressivity and generalization capabilities of MAD policies.\footnote{Our PyTorch implementation is available at: \url{https://github.com/DecodEPFL/mad-rl-policy}.} 

\subsection{Implementation of MAD policies}
Recall the expression \eqref{eq:MAD_policy} of MAD policies. To parametrize the terms $\mathbfcal{M}(\widehat{\mathbf{w}})$ and $\bm{a}(x_0)$ that must be in $\mathcal{L}_p$, we employ linear recurrent units (LRUs) due to their simple structure and rich expressivity~\cite{orvieto2023resurrecting}. To parametrize the direction term $\mathbf{D}(\mathbf{x})$, we employ a multi-layer perceptron  and wrap a $\operatorname{tanh}$ around to guarantee $|D_t(x_{t:0})|<1$. In summary, we parametrize policies of the form
\begin{equation}
\label{eq:MAD_LRU}
    u_t = |\mathtt{LRU}_t(\textcolor{blue}{\theta_1})(\widehat{w}_{t:0})+\textcolor{black}{\mathtt{LRU}_t(\textcolor{blue}{\theta_2})(}x_0\textcolor{black}{)}|\tanh(\textcolor{black}{\mathtt{NN}(}x_t,\textcolor{blue}{\psi}\textcolor{black}{)})\,,
\end{equation}
where the $\mathtt{LRU}_t(\textcolor{blue}{\theta_i})(v_{t:0})$  terms for $i\in \{1,2\}$ are generated by the parametrized system with internal state $\xi_t \in \mathbb{C}^{n_\xi}$:
\begin{align}
\xi_{t+1} &= \textcolor{blue}{\Lambda_i}\,\xi_t + \textcolor{black}{\Gamma}(\textcolor{blue}{\Lambda_i})\textcolor{blue}{\,B_i}\,v_t, \label{LRU_1}\\
\mathtt{LRU}_t(\textcolor{blue}{\theta_i})(v_{t:0})&=\mathtt{NN}(\operatorname{\mathbb{R}e} (\textcolor{blue}{C_i}\,\xi_t ) + \textcolor{blue}{D_i}\,v_t, \textcolor{blue}{\phi_i}) + \textcolor{blue}{F_i}v_t \nonumber\,,
\end{align}
where $\mathbb{R}e$ denotes the real part operator. The input $v_t$ satisfies $v_t = \widehat{w}_t$ for the ``M'' term, and $v_t = x_0$ if $t= 0$ and $v_t = 0$ otherwise for the ``A'' term. In \eqref{LRU_1}, stability is enforced by designing $\textcolor{blue}{\Lambda_i}$ as a diagonal matrix with entries $\lambda_i$ such that $|\lambda_i|<1$, and $\Gamma(\textcolor{blue}{\Lambda_i})$ is a diagonal normalization term.

\subsection{Comparison with DF and standard RL policies}
We benchmark the performance and generalization of MAD policies implemented as per \eqref{eq:MAD_LRU} against DF and standard RL policies. Both MAD and MA policies have around 1900 trainable parameters. The standard RL policy is an MLP with approximately 450 parameters, which appeared to be sufficient to achieve optimized performance. All policies are trained through DDPG to reach a goal position while passing through a narrow corridor and avoiding collisions. With reference to Figure~\ref{fig:state_trajectories}, all initial conditions during training are sampled from the same neighborhood; the orange agent aims to swap from right to left, and vice-versa for the blue agent, forming a crossing ``\hspace{0.04cm}$\mathcal{X}$-shaped'' path. We validate trained policies sampling initial conditions  from the same neighborhoods around those seen in training. To assess generalization capabilities, we interchange the distributions of the initial conditions of the two agents, thus aiming for a non-crossing ``\hspace{0.04cm}\reflectbox{$\mathcal{C}$}\;\hspace{-0.04cm}$\mathcal{C}$-shaped'' path. Figure~\ref{fig:state_trajectories} shows the closed-loop trajectories corresponding to trained MAD policies in the model-based scenario. The left panel displays results on the validation dataset, whereas the right panel shows the successful execution of the task even when the initial conditions of the two agents are swapped. 

In the model-based scenario we train 1) a general MAD policy, and 2) a ``MA'' policy that selects the direction term $D_t(\cdot)$ as in \eqref{eq:d_choice} with no explicit state-feedback, corresponding to a DF policy~\cite{furieri2024learningtoboost}.
In the model-free setting, dynamics are unknown and the disturbance $\widehat{w}_t$ cannot be reconstructed, making DF policies inapplicable. We therefore train an ``AD'' policy of the form $u_t=|\mathtt{LRU}_t(\textcolor{blue}{\theta_2})(x_0)|\operatorname{tanh}(\mathtt{NN}(x_t, {\color{blue} \psi}))$, and compare it with a standard neural policy $u_t = \mathtt{NN}(x_t)$.

\begin{figure}[t]
    \centering
    \includegraphics[width=\columnwidth]{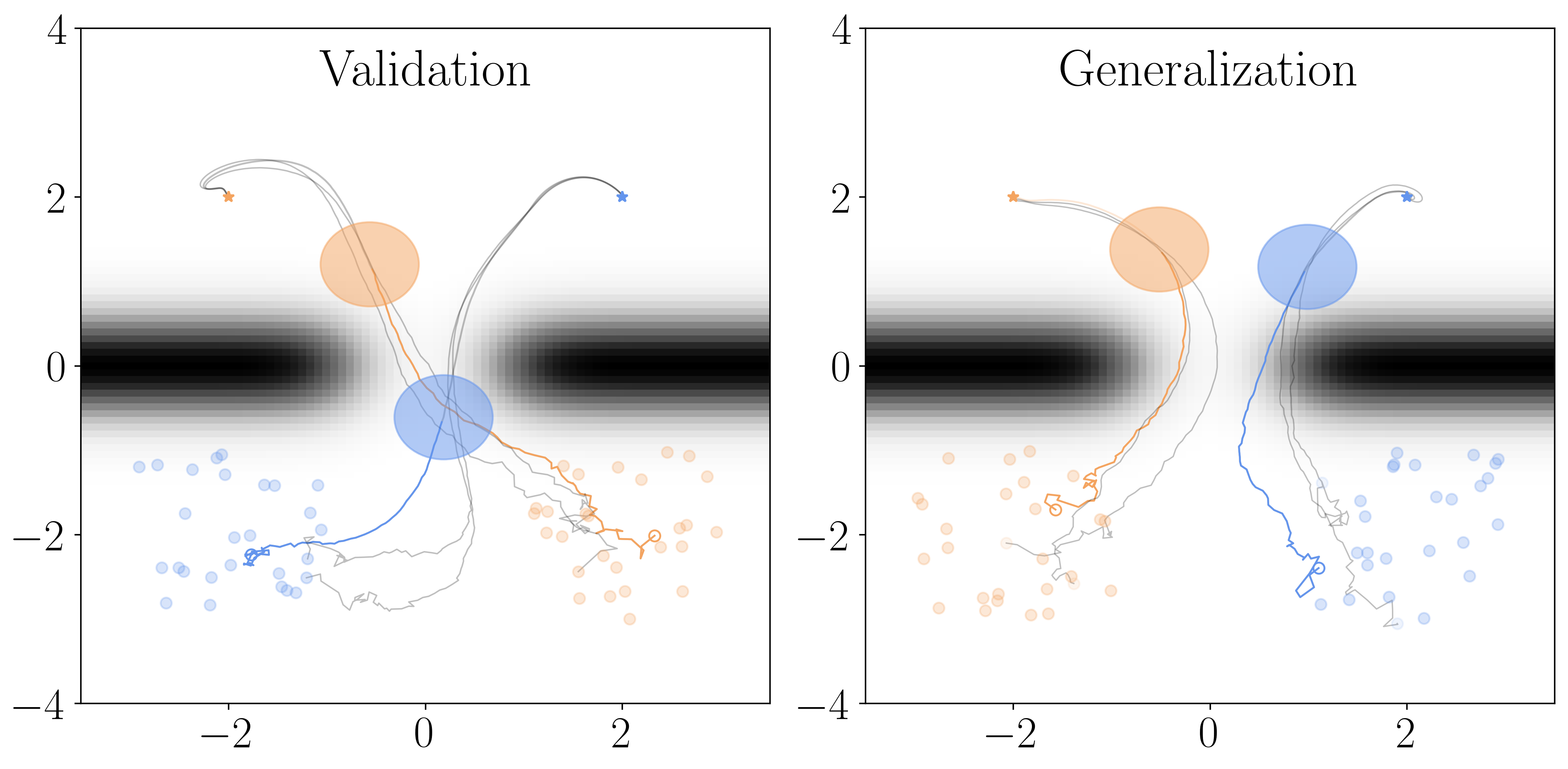}
    \caption{Closed-loop trajectories after training with MAD policies.
    Initial conditions are marked with $\circ$.  The colored balls (and their radii) represent the agents (and their size for collision avoidance). Black objects represent the obstacles.}
    \label{fig:state_trajectories}
\end{figure}

\begin{figure*}[h!]
    \centering
    \subfloat[Validation. \label{fig:validation}]{\includegraphics[height = 4.5cm]{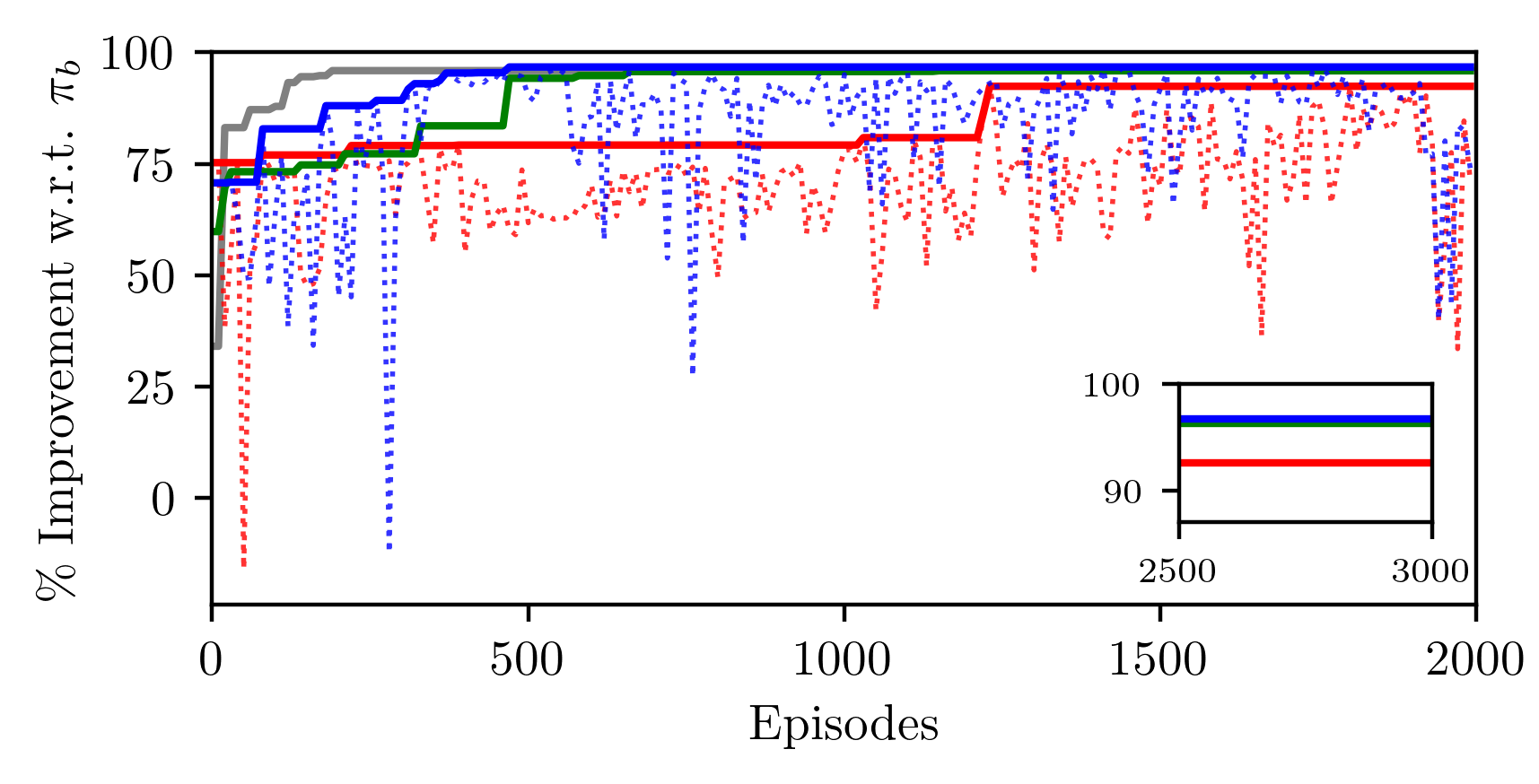}}
    \hfill
    \subfloat[Generalization. \label{fig:generalization}]{\includegraphics[height = 4.5cm]{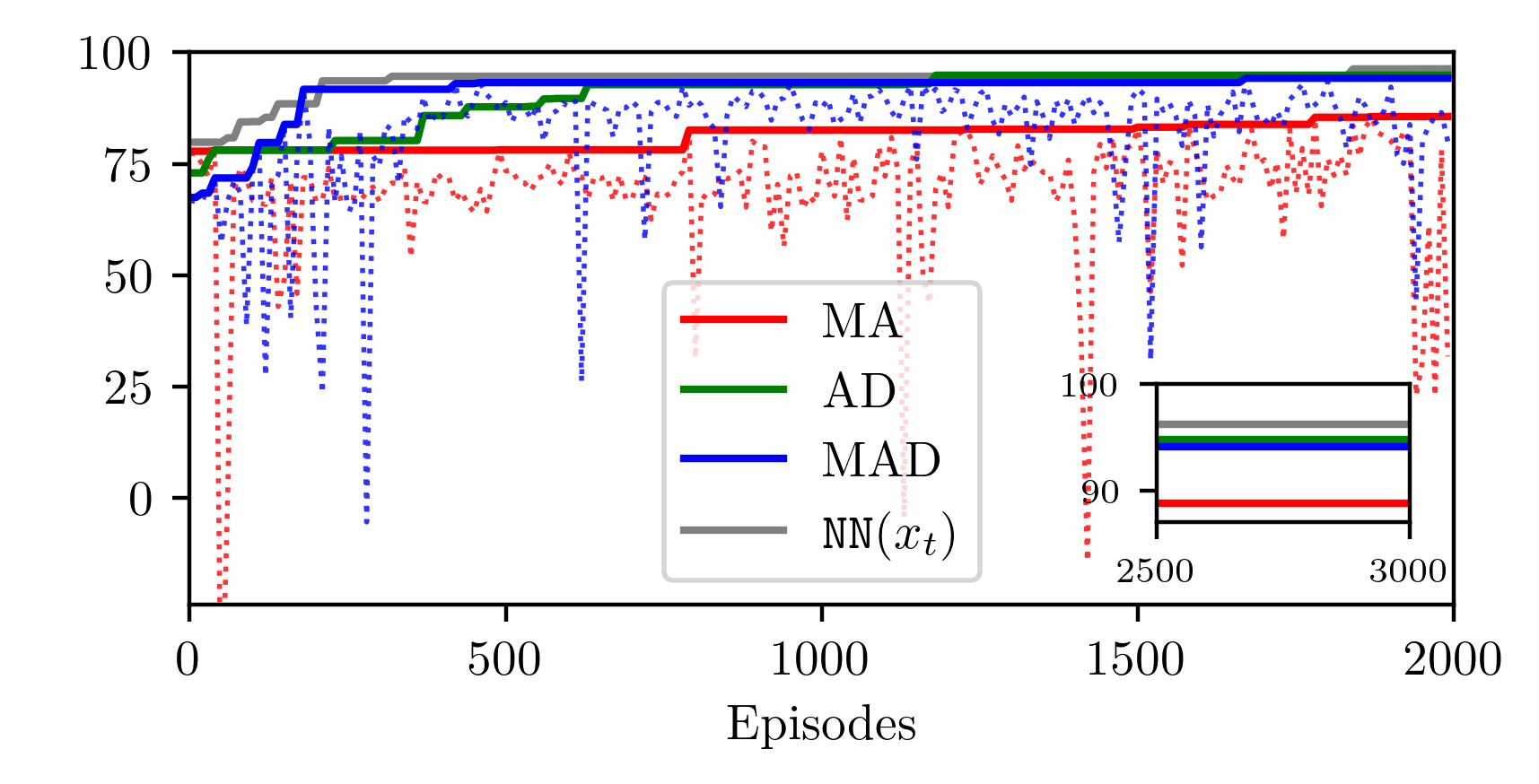}}
  \caption{Percentage improvement in control performance over the pre-stabilizing controller $\pi_b$ defined in \eqref{eq:pre_controller}. Dotted lines represent the performance in each episode, while solid lines indicate the best-so-far performance for each policy class. To reduce visual clutter, episodic performance (dotted lines) is shown only for the MA and MAD policies. 
  The inset plot displays the long-term best-so-far performance between episodes 2500 and 3000.}
\label{fig:training_loss_comparison}
\end{figure*}
In Figure~\ref{fig:training_loss_comparison}, we contrast the performance of the considered policies on the validation (Figure~\ref{fig:validation}) and generalization (Figure~\ref{fig:generalization}) tasks as the training episodes increase. We first note that MAD and AD policies appear to significantly improve the performance on the generalization task as compared to DF policies. In accordance with MLP policies, such an increase in generalization capabilities can be attributed to the inclusion of a neural state-feedback through the direction term. Second, embedding information on the system dynamics in the MAD policy enables faster training with respect to AD, as expected. Crucially, this example shows that enforcing stability constraints via MAD policies does not significantly limit expressivity, as both MAD and AD policies closely match the validation and generalization performance of a standard MLP policy.%

\section{Conclusion}
This paper introduces a novel class of policies tailored to learning stable closed-loop behaviors through reinforcement. The proposed approach leverages a polar decomposition of the control input into magnitude and direction components, extending the expressivity of DF policies without compromising their inherent stability guarantees. 
We showed that adding an explicit neural state-feedback term achieves generalization capabilities that match those of standard RL policies, which, however, may not comply with key stability requirements. Future work includes testing MAD in more complex environments and adapting it to partially observed systems with stronger stability requirements.


\bibliographystyle{IEEEtran}
\bibliography{biblio}   
\newpage

\appendix
\subsection{Proof of Theorem~\ref{th:extended-param}}
\label{app:proof_theorem}
    We first show that $\mathcal{B}_{\mathcal{U}}(\mathbf{w}) \subseteq  \mathcal{B}^{\mathtt{MAD}}_{\mathcal{U}}(\mathbf{w})$. Take any $\mathbf{u} \in \mathcal{B}_{\mathcal{U}}(\mathbf{w})$. It holds that $\mathbf{u} = \mathbfcal{M}(\mathbf{w})$ for some $\mathbfcal{M} \in \mathcal{U}$. By selecting 
    \begin{equation}
    \label{eq:d_choice}
    D_t(x_{t:0}) = \begin{cases}\frac{\mathcal{M}_t(w_{t:0})}{|\mathcal{M}_t(w_{t:0})|}& \text{if $\mathcal{M}_t(\cdot) \neq 0$}\,,\\ 0&\text{otherwise}\,,\end{cases}
    \end{equation}
 where we recursively define $w_t= x_t-f(x_{t-1},u_{t-1})$, and $a_t(x_0) = 0$,  we obtain that $u_t$ can be equivalently written as per \eqref{eq:MAD_policy}, see also \cite{martin2024learning}. Therefore, $\mathbf{u}\in \mathcal{B}^{\mathtt{MAD}}_{\mathcal{U}}(\mathbf{w})$. To show that $\mathcal{B}^{\mathtt{MAD}}_{\mathcal{U}}(\mathbf{w}) \in \ell_p$ when $\mathcal{U} \in \mathcal{L}_p$ and $\mathbf{w} \in \ell_p$, consider the signal $|\mathbfcal{M}(\mathbf{w})+\bm{a}(x_0)|$, where $|\cdot|$ is applied to each time component individually. Since $\mathbfcal{M}(\mathbf{w}) \in \ell_p^m$ and $\bm{a}(x_0)\in \ell_p^m$, their sum also lies in $\ell_p^m$. Further, $|\mathbfcal{M}(\mathbf{w})+\bm{a}(x_0)|\in \ell_p^1$ because any norm on a one-dimensional vector space is equivalent (up to scaling) to the absolute value norm. Since $|u_t|\leq |\mathcal{M}_t(w_{t:0})+a_t(x_0)|$ due to $|D_t(x_{t:0})|\leq 1$, we conclude that $||\mathbf{u}||_{p}\leq ||~|\mathbfcal{M}(\mathbf{w})+\bm{a}(x_0)|~||_p$. This implies that $\mathbf{u} \in \ell_p^m$.

Last, given any $\mathbfcal{M}\in\mathcal{L}_p$, we can rewrite the control input as 
$u_t = |\mathcal{M}_t(w_{t:0})+a_t(x_0)|\,D_t(x_{t:0})$ 
by selecting $D_t$ as per \eqref{eq:d_choice}
so that $|D_t(x_{t:0})|\le 1$ and $a_t(x_0)=0$. This shows that any policy of the form $\mathbf{u}=\mathbfcal{M}(\mathbf{w})$ with $\mathbfcal{M}\in\mathcal{L}_p$ can be expressed as a MAD policy. Moreover, by Theorem~\ref{thm:neurSLS}, for any stabilizing policy $\mathbf{K}$, there exists $\mathbfcal{M}\in\mathcal{L}_p$ whose response $\mathbf{u}=\mathbfcal{M}(\mathbf{w})$ reproduces the same closed-loop behavior. Hence every stabilizing policy can be viewed as a MAD policy with $\mathbfcal{M}=\bm{\Psi}^u$, where $\bm{\Psi}^u:\mathbf{w}\mapsto\mathbf{u}$ is the map induced by $\mathbf{K}$. Conversely, any MAD policy that uses a magnitude operator $\mathbfcal{M}\in\mathcal{L}_p$ and $\bm{a}:\mathbb{R}^n\mapsto \ell_p$ in $\mathcal{L}_p$  with $|D_t(\cdot)|\leq 1$ ensures $\mathbf{u}\in\ell_p$ with the same reasoning as above, thus preserving closed-loop stability.\looseness-1

\subsection{Proof of Proposition~\ref{th:result_robust}}
\label{app:proof_prop}
The proof adapts arguments from \cite{furieri2024learningtoboost} for DF policies. It was proven in \cite[Theorem~3]{furieri2024learningtoboost} that $\widehat{\mathbf{w}} = \mathbf{\Delta}(\mathbfcal{F}(\mathbf{u},\mathbf{w}),\mathbf{u})+ \mathbf{w}$, which leads to
     \begin{equation}
     \label{eq:inequality}
         \norm{\widehat{\mathbf{w}}}_p\leq\gamma(\mathbf{\Delta})(\gamma(\mathbfcal{F})\norm{\mathbf{w}}_p+\gamma(\mathbfcal{F}) \norm{\mathbf{u}}_p+\norm{\mathbf{u}}_p)+\norm{\mathbf{w}}_p\,.
     \end{equation}

Since it holds that $|u_t|\leq |\mathcal{M}_t(\widehat{w}_{t:0})+a_t(x_0)|$ we deduce that $|u_t|^p\leq |\mathcal{M}_t(\widehat{w}_{t:0})+a_t(x_0)|^p $ for any $p\geq 0$ and therefore $\norm{\mathbf{u}}_p^p\leq \norm{\mathbfcal{M}(\widehat{\mathbf{w}})+\bm{a}(x_0)}_p^p$, which also implies $\norm{\mathbf{u}}_p\leq \norm{\mathbfcal{M}(\widehat{\mathbf{w}})+\bm{a}(x_0)}_p\leq \gamma(\mathbfcal{M})\norm{\widehat{\mathbf{w}}}_p+\norm{\bm{a}(x_0)}_p $. By plugging the latter into \eqref{eq:inequality}, we obtain
\begin{align*}
    &(1-\gamma(\bm{\Delta}) \gamma(\mathbfcal{M})\left(\gamma(\mathbfcal{F})+1\right))\norm{\hatbf{w}}\leq\\
    &\leq (\gamma(\bm{\Delta})\gamma(\mathbfcal{F})+1)\norm{\mathbf{w}}+\gamma(\bm{\Delta})(\gamma(\mathbfcal{F})+1)\norm{\bm{a}(x_0)}\,.
\end{align*}
By requiring \eqref{eq:condition_robustness}, we conclude that $\widehat{\mathbf{w}} \in \ell_p$, and hence $\mathbf{u} \in \ell_p$ and $\mathbf{x} \in \ell_p$ because $\mathbf{u}=\mathbfcal{M}(\widehat{\mathbf{w}})$ and $\mathbf{x} = \mathbfcal{F}(\mathbf{u},\mathbf{w})$.

\subsection{The corridor environment}
\label{app:mountains}
In all the examples, we consider two point-mass vehicles, each with position $p_t^{[i]} \in \mathbb{R}^2$ and velocity $q_t^{[i]} \in \mathbb{R}^2$, for $i=1,2$, subject to nonlinear drag forces (e.g., air or water resistance). The discrete-time model for vehicle $i$ is 
\begin{equation}
\label{eq:mechanical_system}
    \begin{bmatrix}
        p_{t}^{[i]}\\
        q_{t}^{[i]}
    \end{bmatrix}
    = 
    \begin{bmatrix}
        p_{t-1}^{[i]}\\
        q_{t-1}^{[i]}
    \end{bmatrix} 
    + T_s
    \begin{bmatrix}
        q_{t-1}^{[i]}\\
        (m^{[i]})^{-1}\left(-C(q_{t-1}^{[i]}) + F_{t-1}^{[i]}\right)\end{bmatrix}+w_t\,,
\end{equation}
where $m^{[i]}>0$ is the mass, $F_t^{[i]} \in \mathbb{R}^2$ denotes the force control input, $T_s>0$ is the sampling time and $C^{[i]}:\mathbb{R}^2\rightarrow \mathbb{R}^2$ is a \emph{drag function} given by
$C^{[i]}(s) = b_1^{[i]} s - b_2^{[i]} \tanh(s)$, for some $0<b_2^{[i]}<b_1^{[i]}$.
Each vehicle 
must reach a target position $\overline{p}^{[i]}\in \mathbb{R}^2$ with zero velocity in a stable way. This elementary goal can be achieved by using a base proportional controller $\pi_b$ that sets
\begin{equation}
\label{eq:pre_controller}
    F'^{[i]}_t = {K'}^{[i]}(\bar p^{[i]}-p_{t}^{[i]})\,,
\end{equation}
with $K'^{[i]} = \operatorname{diag}(k_1^{[i]},k_2^{[i]})$ and $k_1^{[i]},k_2^{[i]}>0$. The overall dynamics $f(x_{t-1},u_{t-1})$ in \eqref{eq:nonlinsys} is given by \eqref{eq:mechanical_system}-\eqref{eq:pre_controller} with $F^{[i]}_t = F'^{[i]}_t + u_t^{[i]}$, 
where $x_t = (p_t^{[1]},q_t^{[1]},p_t^{[2]},q_t^{[2]})$ and 
$u_t = (u_t^{[1]},u_t^{[2]})$ is a performance-boosting control input to be designed.  
As per~\eqref{eq:nonlinsys}, we consider additive disturbances $w_t$ affecting the system dynamics.
Thanks to the use of the prestabilizing controller \eqref{eq:pre_controller}, one can show that $\mathbfcal{F}(\mathbf{u},\mathbf{w})\in \mathcal{L}_2$. We consider the scenario $\mathtt{corridor}$ of \cite{furierineuralSLS} where
each vehicle must reach the target position in a stable way while avoiding collisions between themselves and with two black obstacles (see Figure~\ref{fig:state_trajectories}). 
Each agent is represented with a circle that indicates its radius for the collision avoidance specifications.
When using the base controller $\pi_b$, 
the vehicles successfully achieve the target, however, they do so with poor performance since collisions are not avoided.
To improve performance, we design the input $\mathbf{u}$ by minimizing an infinite-horizon discounted loss function $L^\infty(\mathbf{x},\mathbf{u}) = \sum_{t=0}^{\infty} \alpha^t l(x_t, u_t)$, where the stage-wise loss function is defined as 
\begin{equation*}
    l(x_t,u_t) = l_{\text{traj}}(x_t, u_t) + l_{\text{ca}}(x_t) + l_{\text{obs}}(x_t).
\end{equation*}
The first component of the stage-wise loss, $l_{\text{traj}}(x,u) =\begin{bmatrix}(x-\bar{x})^\top \ u^\top\end{bmatrix} S \begin{bmatrix}(x-\bar{x})^\top \ u^\top\end{bmatrix}^\top$ with $S \succeq 0$, penalizes deviations from target states and excessive control efforts. The second component, $l_{\text{ca}}(x)=S_{ca}\sum_{i \neq j}\frac{\mathbb{I}(|p_t^{[i]} - p_t^{[j]}|^2 < d_{\text{min}}^2)}{|p_t^{[i]} - p_t^{[j]}|^2 + \epsilon}$, with $\epsilon >0$ and $S_{ca} > 0$ penalizes potential collisions between vehicles using an indicator function $\mathbb{I}(\cdot)$ that equals 1 when vehicles $i$ and $j$ are closer than a predefined safe distance $d_{\text{min}} \geq 0$ and 0 otherwise. The final component, $l_{\text{obs}}(x_t)=S_{obs}\sum_{i}\exp\left(-\frac{1}{2}(p_t^{[i]} - \mu_{\text{obs}})^\top \Sigma_{\text{obs}}^{-1}(p_t^{[i]} - \mu_{\text{obs}})\right)$, with $S_{obs}>0$ penalizes proximity to obstacles modeled by Gaussian distributions with mean $\mu_{\text{obs}}$ and covariance $\Sigma_{\text{obs}}$.

\end{document}